\documentclass[12pt]{article}

\usepackage{epsfig}

\usepackage{amssymb}

\usepackage{color}

\def\be{\begin{equation}}
\def\ee{\end{equation}}
\def\ba{\begin{array}{c}}
\def\ea{\end{array}}

\newcommand{\bea}{\begin{eqnarray}}
\newcommand{\eea}{\end{eqnarray}}

\newcommand{\kkt}{\kt\!\kt}
\newcommand{\pbr}{\prec\!}
\newcommand{\pkt}{\!\!\succ\,\,}
\newcommand{\kt}{\rangle}
\newcommand{\br}{\langle}

\newtheorem{thm}{Theorem}

\newtheorem{lemma}[thm]{Lemma}

\newenvironment{proof}{\noindent {\bf Proof.}}{\hfill$\square$\vspace{3mm}\endtrivlist}

\begin{document}


\begin{center}

{\Large \bf

Mutual compatibility/incompatibility of quasi-Hermitian quantum
observables

  }

\vspace{11mm}

{\bf Miloslav Znojil}$^{1,2,3}$

\vspace{3mm}

\end{center}

$^{1}$ {The Czech Academy of Sciences,
 Nuclear Physics Institute,
 Hlavn\'{\i} 130,
250 68 \v{R}e\v{z}, Czech Republic, {e-mail: znojil@ujf.cas.cz}}


$^{2}$ {Department of Physics, Faculty of
Science, University of Hradec Kr\'{a}lov\'{e}, Rokitansk\'{e}ho 62,
50003 Hradec Kr\'{a}lov\'{e},
 Czech Republic}

$^{3}$ {Institute of System Science, Durban University of Technology,
Durban, South Africa}

\vspace{3mm}

\newpage

\section*{Abstract}

In the framework of
quasi-Hermitian
quantum mechanics the
eligible operators of observables
may be non-Hermitian,  $A_j\neq A_j^\dagger$, $j=1,2, \ldots,K$.
In principle,
the standard probabilistic interpretation of the theory
can be re-established
via a reconstruction of
physical inner-product metric $\Theta\neq I$
guaranteeing the quasi-Hermiticity
$A_j^\dagger \,\Theta=\Theta\,A_j$.
The task is easy
at $K=1$
because there are many eligible metrics $\Theta=\Theta(A_1)$.
In our paper
the next case with $K=2$ is analyzed.
The criteria of the existence of
a shared metric $\Theta=\Theta(A_1,A_2)$
are presented and discussed.

\subsection*{Keywords}
.

quantum mechanics of unitary systems;

quasi-Hermitian representations of observables;

constructions of physical inner product metrics;

the criteria of existence of the shared metric for more observables;

\newpage

\section{Introduction}

It used to be widely accepted 
that a consistent picture of a
unitary
quantum system
can only be based on
our knowledge
of its observable quantities
represented by a set of
operators
$A_1$, $A_2$, \ldots
which are all self-adjoint
\textcolor{black}{in a suitable Hilbert space}
(cf.,
e.g., \cite{Messiah}).
More than thirty years ago, Scholtz,
Geyer and Hahne \cite{Geyer} pointed out
that
such a statement may lead to misunderstandings
because
\textcolor{black}{in some fairly realistic models,}
the operators of observables
could be defined as acting,
     \textcolor{black}{
simultaneously,
in {\em two different\,} Hilbert spaces, viz., in spaces
${\cal H}_{mathematical}$ and
${\cal H}_{physical}$.}

     \textcolor{black}{
By construction, the latter two spaces
have to coincide as the sets of ket vectors $|\psi\kt$.
Still, they may be
unitarily non-equivalent because they may
differ by the use of
the two alternative forms of the respective inner products.
Thus, in order to avoid misunderstandings,
one has to speak, first, about
a maximally user-friendly
inner product which specifies
a manifestly unphysical Hilbert space
${\cal H}_{mathematical}$.
At the same time, one also has to speak about
the second, correct physical inner product
which must be used to specify the other,
presumably user-unfriendly
Hilbert space of states
${\cal H}_{physical}$
which is still the only one which provides their
correct probabilistic interpretation.}

The technical details of the idea may be found in \cite{Geyer}.
A decisive advantage of
keeping the two spaces different
has been found in the possibility of
reduction of the difference to
the mere redefinition of the inner product,
 \be
 \br \psi_a|\psi_b\kt_{mathematical}\ \to \
 ( \psi_a|\psi_b)_{physical} \
 \equiv \ \br \psi_a|\Theta|\psi_b\kt_{mathematical}\,.
 \label{change}
 \ee
The amended (often called physical) inner-product metric
$\Theta \neq I$ should be a positive-definite and bounded
operator
which is self-adjoint
in ${\cal H}_{mathematical}$
and has a bounded inverse \cite{Siegl}.
Under these assumptions,
the ``obligatory''
requirement of Hermiticity of  $A_j$s in ${\cal H}_{physical}$
can be re-expressed,
in the preferred representation space
${\cal H}_{mathematical}$,
as their quasi-Hermiticity \cite{Dieudonne}
{\it alias\,}
$\Theta-$pseudo-Hermiticity \cite{ali},
 \be
 A_j^\dagger\, \Theta=\Theta\,A_j\,,
 \ \ \ \ \forall j\,.
 \label{foral}
 \ee
Needless to add that
the two Hilbert spaces would coincide
in the limit of $\Theta \to I$.
Within the generic theory with $\Theta \neq I$, nevertheless,
the ``hidden Hermiticity''
constraint (\ref{foral})
will still
imply the reality of the spectrum.

In most applications
the model-building process
proceeds in opposite direction.
In a way
sampled in \cite{Geyer} or \cite{BG},
one is just given
a set $\{A_j\}$ of
the candidates for observables.
All of these
operators
have to be treated,  in
${\cal H}_{mathematical}$, as non-Hermitian,
$A_j\neq A_j^\dagger$.
For all of them,
the reality of the spectrum
is merely necessary, but not sufficient, condition
of their observability and
mutual compatibility.
One has to prove the mathematical consistency of the theory
by showing that {\em all\,} of the preselected
operators of observables
{\em share\,} the quasi-Hermiticity property~(\ref{foral}).
A clarification of these items is also the main
purpose of our present paper.

In our preceding study \cite{arabky} of the problem
we have shown that even when one decides to work with the mere doublet of
the ``input'' operators $A_1$ and $A_2$,
the acceptable inner-product metric
need not exist at all. The essence of such a no-go theorem
was that the two ``arbitrary'' operators $A_1$ and $A_2$ had
to be quasi-Hermitian with respect to a single,
subscript-independent metric $\Theta$.
This simply imposes too many constraints upon $\Theta$
in general.

In what follows we will outline a continuation of the latter study.
We will introduce and describe an efficient strategy of
the construction of $\Theta$ whenever it does exist.


\section{Self-adjoint Hamiltonians in quasi-Hermitian form\label{dvojka}}

The
innovative
two-Hilbert-space
reformulation of quantum theory
as offered in paper \cite{Geyer}
may be
called quasi-Hermitian quantum mechanics (QHQM).
In its original formulation, unfortunately,
most of the mathematical
as well as phenomenological roots of such a theory
seemed to be rather specific,
aimed just at nuclear physicists.
Indeed, the authors themselves were nuclear physicists
so that they
decided to
illustrate
the merits
of the formalism
via a rather complicated
many-particle example.
In a broader physics community, therefore, the more abstract
theoretical message
as provided by paper \cite{Geyer}
remained
more or less unnoticed.

In the words of
a recent review paper \cite{ali},
``the main problem with this formalism
is that it is generally very difficult to
implement'' (cf. p. 1217 in {\it loc. cit.}).
A properly amended
version of the idea of QHQM
appeared, fortunately,
just a few years later.

\subsection{Quantum models with single quasi-Hermitian observable}

An ultimate and successful
version of the QHQM approach
to quantum systems has been proposed by
Bender with co-authors \cite{BB,BBJ}.
Their upgrade of the theory
has been simpler, better motivated,
better presented
and much more widely accepted.
Almost immediately,
the modified theory
(currently known
under a nickname of ${\cal PT}-$symmetric
quantum mechanics \cite{Carl})
became truly popular,
partly also due to
the methodically very fortunate
restriction of attention
to specific quantum systems characterized  by the use

\begin{itemize}

\item[(i)] of a {\em single\,} non-Hermitian
operator $A_1$
serving phenomenological purposes and
yielding, typically, bound-state energies as its eigenvalues \cite{book};

\item[(ii)] of the most common
formulation the theory called Schr\"{o}dinger picture
in which the observable
plays the role of quantum Hamiltonian, $A_1=H$ \cite{ali};

\item[(iii)] of the mere ordinary-differential
Hamiltonians $H=-d^2/dx^2+V(x)$ in which even the local
interaction potentials themselves
were kept maximally elementary \cite{Geza};

\item[(iv)] of an auxiliary antilinear parity times time-reversal
symmetry
(${\cal PT}-$symmetry). This made the Hamiltonians tractable
as self-adjoint in an {\it ad hoc\,} Krein
space \cite{Langer,AKbook}.

\end{itemize}

\noindent
All of these well-invented simplifications contributed
to the enormous growth of
popularity of the non-Hermitian representations of
quantum systems exhibiting ${\cal PT}-$symmetry.
What followed was
a highly productive and fruitful
transfer of the related mathematics
to several non-quantum areas of physics (cf., e.g., \cite{Christodoulides}).
Nevertheless, we believe that
it is now time to remove or weaken
at least some of the above-listed simplifications while
turning attention

\begin{itemize}

\item[(i)] to the models using at least two
observables, i.e., not only $A_1$ (Hamiltonian)
but also $A_2$
(in \cite{arabky}, the second operator was interpreted as
representing the coordinate);

\item[(ii)] to the possibilities of
description of systems
in the quasi-Hermitian Schr\"{o}dinger picture
of Ref.~\cite{Faria},
in the quasi-Hermitian interaction picture
of Refs.~\cite{NIP,NIPb},
or in the quasi-Hermitian Heisenberg picture
of Ref.~\cite{HP};

\item[(iii)] to some alternative simplifications
of dynamics using, e.g., the discretizations of
coordinates allowing
the systems to live in a finite-dimensional
Hilbert space ${\cal H}^{(N)}$;

\item[(iv)] to the systems with less symmetries
(and, in particular, without ${\cal PT}-$symmetry) and
to the
constructions making use, at $N<\infty$, of
the methods of linear algebra.

\end{itemize}


\subsection{The concept of Dyson map}

During the history of development of
QHQM
as outlined, say, in  \cite{book} or \cite{Carlbook},
people proposed several simplifications
treating the
inner-product metric $\Theta$ as a product of some other
auxiliary operators.
The first attempt
can be traced back to the Dyson's
paper \cite{Dyson} in which
{\em both\,} of the above-mentioned Hilbert spaces
${\cal H}_{mathematical}$ and ${\cal H}_{physical}$
  \textcolor{black}{
(with their shared ket-vector elements
denoted by the same symbol
$|\psi\kt$)}
were introduced as auxiliary,
complementing just the usual, third
Hilbert space ${\cal H}_{standard}$
of traditional textbooks
    \textcolor{black}{
(here we follow paper \cite{SIGMAdva}
and introduce, for the sake of clarity, the different,
``curved''
ket symbol $|\psi\pkt$
for its elements).}

The essence of the problem was that the
convergence of the routine variational
calculations of
the bound-state energy spectra
of a certain multi-fermioninc quantum system
appeared to be, in the latter
representation space, too slow.
The reason was
the existence of
fermion-fermion correlations.
Having this in mind,
Dyson
managed to accelerate the convergence
via a map
between states
$|\psi\pkt \in {\cal H}_{standard}$ and $
|\psi\kt \in {\cal H}_{mathematical}$,
 \be
 |\psi\pkt =\Omega\,|\psi\kt\,.
 \label{dm}
 \ee
%
Although the mapping was designed to
mimic the correlations,
its conventional unitary form did not work.
A success has only been achieved when
the operators $\Omega$ were chosen non-unitary, i.e., such that
 \be
 \Omega^\dagger \,\Omega = \Theta \neq I\,.
 \label{met}
 \ee
This was precisely the reason why Dyson had to replace the
space
${\cal H}_{mathematical}$
by its amendment
${\cal H}_{physical}$. Indeed,
only the
latter Hilbert space
could have been declared physical,
unitarily equivalent to its standard textbook
alternative. Indeed, the emergent equivalence
appeared to be due the coincidence
of all of the inner products of states
before and after the mapping,
 \be
 \pbr \psi_a|\psi_b\pkt =
 \br \psi_a|\Theta|\psi_b\kt\,.
 \label{adm}
 \ee

In nuclear physics
the Dyson's map (\ref{dm})
found a number of applications \cite{Jensen}, mainly
because the
complicated fermionic statistics
(reflecting the Pauli principle in  ${\cal H}_{standard}$)
was, by construction, simplified and replaced
by its
bosonic-representation
alternative
in both
${\cal H}_{mathematical}$
and
${\cal H}_{physical}$.
This was another benefit which
simplified the
calculations performed
in the
most friendly Hilbert space
${\cal H}_{mathematical}$ {\it alias\,}, briefly,
in ${\cal H}_{}$.
Simultaneously,
with the knowledge of a (preselected)
factor-operator $\Omega$
it was entirely straightforward
to interpret the
manifestly non-Hermitian Hamiltonian $A_1\,\equiv\,H \neq H^\dagger$
as isospectral
with its transform
  \be
  \Omega\,H\,\Omega^{-1}=\mathfrak{h}\,.
 \label{metr}
  \ee
Due to the validity of
the quasi-Hermiticity (\ref{foral}) of $H=A_1$,
the latter operator $\mathfrak{h}$
is
self-adjoint in ${\cal H}_{standard}$ and
isospectral with $H$.

\section{The problem of compatibility}

In 1956,
Dyson
introduced the mapping (\ref{dm}) + (\ref{met})
interpreted as an auxiliary, non-unitary
transformation of correlated
fermions into interacting bosons.
He
managed to represent fermionic
(so called Fock) space ${\cal H}_{standard}$ in unphysical
bosonic ${\cal
H}_{mathematical}$
as well as in the correct bosonic ${\cal
H}_{physical}$. All this was just
the technically motivated
change of representation,
with the difference between the latter two Hilbert spaces
only involving
the
fermion-correlation-reflecting
change of the inner product
in the two interacting-boson representations
of the original fermionic system.

In 1992 the authors of review \cite{Geyer}
extended these considerations and
concluded that
under
certain conditions,
the role of the observables of a
conventional quantum system
could be played by a set
$A_j$
of manifestly non-Hermitian but bounded operators possessing
real and discrete spectra.
For illustration
they choose a system
characterized by a triplet $A_1$, $A_2$ and $A_3$
of certain linearly independent
preselected non-Hermitian
candidates for the observables
(cf. section Nr. 3 in {\it loc. cit.}).
Such a choice was
truly impressive
because it admitted even
the emergence of quantum phase transitions.
Later, as we already mentioned,
such a direction of research has been
found difficult and, more or less completely,
abandoned.
The subsequent developments of the field
led to the preference of the studies of the
less realistic quantum models, with their dynamics
controlled by the mere single non-Hermitian observable.

\subsection{Quantum models with two quasi-Hermitian observables}

Until recently, most of the
applications
of the QHQM approach
remained restricted
to the systems characterized by the
observability of the mere single operator $A_1=H$
representing the Hamiltonian.
In our present paper we are going to return
to the origins. We will assume
a more realistic dynamical input
knowledge of more than one non-Hermitian observable,
having in mind the well known fact that
such a knowledge is essential and
needed, after all, even for
an unambiguous reconstruction
of the correct physical metric $\Theta$.

In 2017 we already initiated such a return
to the origins in \cite{arabky}.
We
managed to show there that
in the general multi-observable setting,
even the first additional linearly independent
operator candidate $A_2$
for a non-Hamiltonian observable
cannot be arbitrary.
In a way inspired by a few specific realistic
examples we demonstrated that
an unrestricted choice of the two linearly independent
operators $A_1$ and $A_2$
{\em need
not\,} be acceptable as
representing a pair of observables
of a single quantum system in general.

For the purposes of the proof of the latter statement
it has been sufficient to use perturbation techniques.
We showed that even in the specific ``almost Hermitian''
models with
two preselected ``almost Hermitian''
candidates $A_1$ and $A_2$ for observables and under
the related small-non-Hermiticity-reflecting ansatz
 \be
 \Theta=I+\epsilon F + {\cal O}(\epsilon^2)
 \label{artif}
 \ee
one arrives at the constraints
(in fact, at the consequences of Eq.~(\ref{foral}))
for which
a nontrivial first-order-metric solution~$F$
need not exist at all.

In our present continuation of the latter study
we will drop
the  ``almost Hermiticity'' assumptions
as over-restrictive and redundant.
We will analyze the
``mutual compatibility/incompatibility'' problem
in its full generality. For the
methodical reasons as well as for the sake of brevity
of our text
we will only assume that
the dimension $N$ of the relevant Hilbert spaces is finite.

\subsection{A new trick: Two auxiliary standard spaces}

In section \ref{dvojka} we worked, in effect, with
a triplet of Hilbert spaces, viz., with ${\cal H}_{standard}$
(known from the conventional textbooks and, in the
above-mentioned Dyson's
example, ``fermionic'') and with ${\cal H}_{mathematical}$ and
${\cal H}_{physical}$ (in the Dyson's
example, both of them were ``bosonic'').
The former one was related to the other two
by the Dyson map $\Omega$ (cf. Eq. (\ref{dm})),
while the ``physicality'' of ${\cal H}_{physical}$
was due to relation (\ref{adm}).

There existed just two versions of the observable
Hamiltonian (cf. Eq.~(\ref{metr})).
More precisely, the representation $\mathfrak{h}$ of the
observable Hamiltonian in ${\cal H}_{standard}$
was, necessarily, diagonalizable.
In some cases,
one could even work, directly, with its most elementary diagonal-matrix
form (cf., e.g., Eq. Nr. 5 in \cite{frantiPLA}).

With the turn of attention to the
quantum models characterized by the
pairs of observables $A_1$ and $A_2$,
the overall theoretical scenario becomes different.
First of all,
once we fix a Dyson map, we may still
consider the pair of analogues
of Eq.~(\ref{metr}), viz., mappings $A_1\,\to\,\mathfrak{a}_1$
and $A_2\,\to\,\mathfrak{a}_2$.
Now, the point is that
we can only diagonalize {\em one\,} of the resulting
operators $\mathfrak{a}_j$.
For this reason, we will rather insist on their diagonality
and achieve this goal by the use of {\em two different\,}
Dyson maps,
 \be
 \Omega_1\,A_1\,\Omega_1^{-1} = \mathfrak{a}_1\,,
 \ \ \ \
 \Omega_2\,A_2\,\Omega_2^{-1} = \mathfrak{a}_2\,
 \label{metternich}
 \ee
This leads to the two different
images of states in ${\cal H}_{standard}$
(cf. Eq. (\ref{dm})) and, therefore,
to the two different pictures of physics
or, after a re-wording,
to the two different versions
of the physical Hilbert space of textbooks.

In the light of the definition
of metric (\ref{met}) and of the
inner-product equivalence (\ref{adm}),
we would also have the two
versions of the
correct
Hilbert space ${\cal H}_{physical}$
which would be, naturally, mutually incompatible.
There is only one way of avoiding the difficulty:
Given the two
diagonalization requirements (\ref{metternich}),
we must
make use of the ``hidden''
ambiguity of the respective inner-product metrics
as explained in our older paper \cite{SIGMAdva}.

\section{Problem of compatibility/incompatibility}

\textcolor{black}{
The assumption of diagonality
of the two above-mentioned
real matrices $\mathfrak{a}_1$
(with non-vanishing
elements $(\mathfrak{a}_1)_{nn}=a_n^{[1]}$) and $\mathfrak{a}_2$
(with $(\mathfrak{a}_2)_{nn}=a_n^{[2]}$)
enables us
to replace Eq.~(\ref{metternich})
by its equivalent
Hermitian-conjugate matrix-equation
alternative
 \be
 A_1^\dagger\,\Omega_1^{\dagger} = \Omega_1^\dagger\,\mathfrak{a}_1\,,
 \ \ \ \
 A_2^\dagger\,\Omega_2^{\dagger} = \Omega_2^\dagger\,\mathfrak{a}_2\,.
 \label{mmetternich}
 \ee
Moreover, we may treat both of the
matrices $\Omega_j^\dagger$ with $j=1$ or $j=2$
as concatenations of eigenvectors of the respective
matrices $A_j^\dagger$.
Hence, once we accept the notation convention
of paper \cite{SIGMAdva} and  once we
mark these eigenvectors by a doubled
ket symbol, i.e., once we replace $\kt$ by $\kkt$,
we will be able to re-read
every column of the two Eqs.~(\ref{mmetternich})
as a relation
which has the
very standard form of a Schr\"{o}dinger-equation-like
eigenvalue problem,
  \be
 A_1^\dagger\, |\psi_m^{[1]}\kkt =  |\psi_m^{[1]}\kkt\,{a}_m^{[1]}\,,
 \ \ \ \ \ \ \
 A_2^\dagger\, |\psi_m^{[2]}\kkt =  |\psi_m^{[2]}\kkt\,{a}_m^{[2]}\,,
 \ \ \ \ \ \ m=1,2,\ldots,N
 \,.
 \label{ternich}
 \ee
As a consequence,
both of the concatenations
of the column vectors will have the
obvious form
composed of the eigenvectors
obtained by the solution of Eqs.~(\ref{ternich}),
 $$
 \Omega_j^\dagger
 =
 \{
 |\psi_1^{[j]}\kkt,|\psi_2^{[j]}\kkt,\ldots,|\psi_N^{[j]}\kkt
 \}\,,
 \ \ \ \ \ j=1,2
 \,.
 $$
Once
we decided to deal, in the present paper,
just with the pairs of observables
$A_1$ and $A_2$,
the latter ``ketket'' eigenvectors
had to form just the two different families
of eigenvectors
of
$A_j^\dagger$ with $j=1$ or with $j=2$
in general. What is essential is that they
may be calculated by stanard methods, and that they
may be then used to define the metric via Eq.~(\ref{met}).
}

\subsection{Residual freedom}

Once we are given the two
dynamics-representing operators $A_1$ and $A_2$
which are defined as non-Hermitian in
the finite-dimensional Hilbert space ${\cal H}_{mathematical}^{(N)}$,
relations (\ref{metternich})
may be  re-written, equivalently, as the
two entirely conventional matrix eigenvalue problems (\ref{ternich}).
The eigenvectors form here the columns of
matrices $\Omega_j^{\dagger}$.
Standard methods may be used to reconstruct
both
the latter two matrices
and the corresponding diagonal matrices $\mathfrak{a}_j$
of
eigenvalues which are all real
and, say, non-degenerate.

Now, our most important
observation is that
every (i.e., the $n-$th) column of
the benchmark solutions $\Omega_1^{\dagger}$
and/or $\Omega_2^{\dagger}$
of the two
respective eigenvalue problems (\ref{ternich})
is only defined up to an arbitrary
(i.e., real or complex) non-vanishing
multiplication factor $\kappa_1^{[n]} \neq 0$
or $\kappa_2^{[n]} \neq 0$, respectively.
Thus, once we keep these factors just real and positive, and
once we let these factors
form the two respective diagonal matrices
$\mathfrak{c}_1$
and
$\mathfrak{c}_1$,
we reveal that
every initial, arbitrarily normalized solution
matrix $\Omega_j^\dagger$
is not unique
and can be modified as follows,
 \be
 \Omega_j^\dagger \ \to \ \Omega_j^\dagger \, \mathfrak{c}_j\,,\ \ \ j = 1, 2\,.
 \ee
Using such a transformation, therefore,
one obtains all of the acceptable solutions
which are parametrized by the respective
real, positive and diagonal matrices $\mathfrak{c}_j$.
The $n-$th column of our arbitrarily normalized and
fixed initial solution $\Omega_j^\dagger $
of the respective Eq.~(\ref{ternich})
is being multiplied by an entirely arbitrary subscript-dependent
constant
$\kappa^{[n]}_j=\left (\mathfrak{c}_j\right )_{n,n}\neq 0$.
These values
may even be complex because, in a way pointed out in \cite{SIGMAdva},
there is no constraint imposed upon the variability
of the diagonal elements of $\mathfrak{c}_j$ with $ j = 1, 2$.
Here,
without any loss of generality we will
use just the real and positive values of these parameters.

Similarly, also the eligible metrics
defined by formula (\ref{met})
become, for the same reason, ambiguous,
 \be
 \Theta_j  = \Theta_j (I)=\Omega_j^\dagger\,\Omega_j
 \ \to \
 \Theta_j = \Theta_j (\mathfrak{c}_j)=\Omega_j^\dagger\,\mathfrak{c}_j^2
 \,\Omega_j
 \,,\ \ \ j = 1, 2\,.
 \label{[11]}
 \ee
At the same time, as long as such a characterization
of the ambiguity is exhaustive,
we may now formulate our present main conclusion.

\begin{lemma}
The two $N$ by $N$ matrix candidates $A_1$ and $A_2$ for observables
can be declared compatible if and only if
there exist positive diagonal matrices $\mathfrak{c}_1$
and $\mathfrak{c}_2$ such that
 \be
 \Theta_1 (\mathfrak{c}_1)=\Theta_2 (\mathfrak{c}_2)\,.
 \label{veri}
 \ee
\end{lemma}

\subsection{Criteria}

In a way not noticed in \cite{arabky},
the test and verification of coincidence (\ref{veri})
becomes facilitated by the solution of Eqs.~(\ref{ternich}).
This enables us to rewrite Eq.~(\ref{veri})
as relation
 \be
 \Omega_1^\dagger\,\mathfrak{c}_1^2
 \,\Omega_1=\Omega_2^\dagger\,\mathfrak{c}_2^2
 \,\Omega_2\,
 \label{2veri}
 \ee
in which one can only vary
the $2N$ diagonal
matrix elements of
$\mathfrak{c}_1^2$ and $\mathfrak{c}_2^2$.

One of the consequences of the
set of constraints (\ref{2veri})
is that we can immediately eliminate
half of the variable parameters since we can write, say,
 \be
 \mathfrak{c}_2^2=M^\dagger\,\mathfrak{c}_1^2
 \,M\,,
 \ \ \ \ M=
 \Omega_1\,
 \Omega_2^{-1}\,.
 \label{3veri}
 \ee
Both sides of this relation
are Hermitian matrices which
have two parts. The diagonal
parts are elementary and
provide simply an explicit definition
of the diagonal matrix $\mathfrak{c}_2^2$ in terms
of the known matrix $M$ and the
variable diagonal matrix $\mathfrak{c}_1^2$.

The $N-$plet of the
real and positive variables in $\mathfrak{c}_1^2$
is then still constrained by the
rest of relation (\ref{3veri}) which
can be read, at the larger $N$, as an
over-determined
set of $N(N-1)/2$ constraints.
They may be written, say,
as a multiplet of linear
equations
 \be
 \left (
M^\dagger\,\mathfrak{c}_1^2
 \,M
 \right )_{m,n}=0\,, \ \ \ \ \forall m<n\,.
 \label{[15]}
 \ee
In practice, the implementation of
such a criterion
would
\textcolor{black}{be only straightforward
in the two-dimensional case since at
$N=2$ we would just have a single constraint (\ref{[15]})
imposed upon two real parameters. In general,
in contrast, the two arbitrarily chosen
independent matrices $A_1$
and $A_2$ will not be mutually compatible at $N\geq 4$
(and the more so when $N = \infty$)
since $N(N-1)/2 > N$. In these cases one will}
have to rely on the purely numerical techniques.

For some purposes, a less pragmatic
alternative approach to the problem could be
based on our following final result.

\begin{lemma} \label{lemmadva}
  The two $N$ by $N$ matrix candidates $A_1$ and $A_2$ for observables
can be declared compatible (i.e., simultaneously quasi-Hermitian)
if and only if the $2N-$parametric
matrix
 \be
 {\cal U}(\mathfrak{c}_1,\mathfrak{c}_2)=\mathfrak{c}_1\,M\,
 \mathfrak{c}_2^{-1}
 \label{[16]}
 \ee
is unitary.
\end{lemma}
\begin{proof}
The unitarity
is equivalent to requirement~(\ref{3veri}).
\end{proof}

 \noindent
One of the consequences of the latter observation
would be constructive.
At a given Hilbert-space dimension $N$, indeed,
one could choose an arbitrary unitary matrix ${\cal U}$
and, having recalled Eq.~(\ref{[16]}), one would be able to
define the matrix $M$.
After its
pre-multiplication by some $2N$ arbitrary constants,
a more or less arbitrary factorization
would yield,
in the light of Eq.~(\ref{3veri}),
the doublet of matrices $\Omega_1$ and $\Omega_2$.
In terms of any one of them
one could define the ``shared'' metric $\Theta$.
An arbitrary choice of the
two diagonal matrices $\mathfrak{a}_1$
and $\mathfrak{a}_2$ would then finally lead to the
construction of the compatible pair of operators
$A_1$ and $A_2$ via the respective relations (\ref{metternich}).

In an opposite direction
of the constructive analysis, the process of
implementation of the
unitarity criterion (\ref{[16]})
could be
based
on the most general unitary-matrix ansatz ${\cal U}$ for the
left-hand-side of Eq.~(\ref{[16]}).
Thus, in the simplest example with $N=2$
such an ansatz would be four-parametric,
 \be
 {\cal U}=\exp \left [
 \begin{array}{cc}
 i\alpha&i\gamma+\delta\\
 i\gamma-\delta&i\beta
 \ea
 \right ]\,.
 \ee
This means that at any given two-by-two
dynamical-input-information matrix $M^{(2)}$ one could
search for the four unknown parameters
$\kappa_j^{[n]}$ with $j=1,2$ and $n=1,2$
via the solution of the set of constraints (\ref{[16]})
forming a quadruplet of coupled nonlinear
algebraic equations.

\textcolor{black}{
\section{Example}
}


\textcolor{black}{
One of
the expected
outcomes of our present methodical study can be seen in its possible
relevance in the theory of quantum gravity.
In this context, indeed,
many conventional quantization
methods
cease to be applicable
because they rely upon the
availability of a classical ``background''
represented, typically, by the space or space-time.
In the case of gravitational field, in contrast,
one also has to quantize the space-time itself,
in principle at least \cite{Rovelli}.
As a consequence,
a consistent quantum theory of gravity must be,
in Ashtekar's words, ``background-independent'' \cite{Ashtekar}.
In this sense, a minimal number of
the candidates for observables $A_j$ is two,
and the demonstration of their mathematical compatibility
is of a deep and fundamental relevance, indeed.
}

\textcolor{black}{
In a narrower setting of quantum cosmology, in addition,
one has to quantize classical scenarios in which
one often encounters several
different forms of singularities.
Thus, typically, the
latter candidates $A_j$
have to admit the existence of the
singularities
represented, in the Hilbert-space operator context,
by the so called Kato's \cite{Kato} exceptional points (EPs).
In a realistic setting these operators even
have to be non-stationary and non-Hermitian
and, in parallel, also the physical Hilbert-space inner-product metric itself
has to be non-stationary and non-Hermitian.
}

\textcolor{black}{
In spite of all of these technical challenges,
the practical applicability of the general theory
has still been demonstrated and illustrated, in the existing literature,
by a number of simple physical examples.
A few particularly suitable examples of such a type
were described in Ref.~\cite{passage}.}

\textcolor{black}{
From the purely phenomenological point of view,
one of the important merits of the latter classes of models
is that all of them admit the non-Hermitian exceptional-point
degeneracy \cite{Kato}. This means that
all of them could
be used as mimicking the
Big-Bang-like singularity
in quantum cosmology.
Their aadditional methodical merit is
that all of these families
of solvable models
are represented by matrices $A_j$ of arbitrary dimensions $N$.
}

\textcolor{black}{
As we already mentioned above, nevertheless,
we do not need
the large-matrix models
(because
the conditions of their mutual compatibility become
over-restrictive
at $N\geq 4$)
nor the minimal models with $N=2$
(for them the parametric freedom survives).
Thus, we decided to pick up,
for our present illustration purposes,
just the one-parametric sets of the
$N=3$ matrices such that
 \be
 {A_1}^\dagger=
 \left[ \begin {array}{ccc} -2&\sqrt {2-2\,{s}^{2}}&0\\
 \noalign{\medskip}-\sqrt {2-2\,{s}^{2}}&0&\sqrt {2-2\,{s}^{2}}
 \\
 \noalign{\medskip}0&-\sqrt {2-2\,{s}^{2}}&2
 \end {array} \right]
 \label{asyre}
 \ee
and
 \be
 {A_2}^\dagger=
 \left[ \begin {array}{ccc}
  -2\,ia&\sqrt {2}&0\\\noalign{\medskip}
\sqrt {2}&0&\sqrt {2}\\\noalign{\medskip}0&\sqrt {2}&2\,ia\end {array}
 \right]\,.
 \label{cosy}
 \ee
They seem totally different
(notice that the former one is asymmetric but real while the
latter one is complex symmetric).
Moreover, as long as we can quickly determine their
diagonal partners
of Eq.~(\ref{metternich}), say,
 \be
 {\bf a_1}=
\left[ \begin {array}{ccc}
 0&0&0\\\noalign{\medskip}0&-2\,s&0\\
 \noalign{\medskip}0&0&2\,s\end {array} \right]
 \label{asyredi}
 \ee
and
 \be
 {\bf a_2}=
 \left[ \begin {array}{ccc}
  0&0&0\\
 \noalign{\medskip}0&-2\,\sqrt {-{a}^{2}+1}&0\\
 \noalign{\medskip}0&0&2\,\sqrt {-{a}^{2}+1}\end {array}
 \right]
 \label{cosydi}
 \ee
we may also find a closed one-parametric forms
of the arbitrarily normalized related matrices
 \be
 \Omega_1^\dagger=
 \left[ \begin {array}{ccc}
  -2+2\,{s}^{2}&2\,s+1+{s}^{2}&-2\,s+1+{s}^{2}\\
 \noalign{\medskip}-2\,\sqrt {2-2\,{s}^{2}}&\sqrt {2-2\,{s}^{2}}
 \left( s+1 \right) &-\sqrt {2-2\,{s}^{2}} \left( -1+s \right)
\\\noalign{\medskip}-2+2\,{s}^{2}&1-{s}^{2}&1-{s}^{2}\end {array}
 \right]
 \label{asyreQ}
 \ee
plus, similarly, also matrix $\Omega_2^\dagger$
(which appeared, incidentally, just a bit too large for its display in print).
 }

 \textcolor{black}{
In the next step of illustration we will fix the parameters $s=1/2$
and $a=1/2$ yielding a simpler formula for
 \be
 \Omega_1^\dagger=
 \left[ \begin {array}{ccc} -3/2&9/4&1/4\\
 \noalign{\medskip}-\sqrt {6}&3/4\,\sqrt {6}&1/4\,\sqrt {6}\\
 \noalign{\medskip}-3/2&3/4&3/4
\end {array} \right]
 \label{asyreQb}
 \ee
as well as a more comfortably printable
formula for
  \be
 {\Omega_2}^\dagger=  
\left[ \begin {array}{ccc} 3/2&3/8+3/8\,i\sqrt {3}&3/8-3/8\,i\sqrt {3}\\
\noalign{\medskip}3/4\,i\sqrt {2}&-3/8\,\sqrt {2}\sqrt {3}-3/8\,i
\sqrt {2}&3/8\,\sqrt {2}\sqrt {3}-3/8\,i\sqrt {2}\\\noalign{\medskip}-
3/2&3/4&3/4\end {array} \right]\,.
 \label{cosyQ}
 \ee
What remains to be done is just the
reconstruction of the two entirely general alternative
three-parametric metrics of Eq.~(\ref{veri}) or (\ref{2veri}).
With abbreviations $\left (\mathfrak{c}_1\right )_{nn}^2=x_n$
and $\left (\mathfrak{c}_2\right )_{nn}^2=y_n$
these $N$ by $N$ identities are tractable as linear
equations which have,
ultimately, the following
explicit
solution
defined in terms of a single variable $\varrho$, viz.,
 $$
   {\it x_1}=-1/12\,{\it \varrho}+1/12\,i\sqrt {3}{\it \varrho}\,,
 $$
 $$
   {\it x_2}=1/9\,{\it \varrho} \,,
 $$
 $$
   {\it x_3}={\it \varrho}\,,
 $$
 $$
   {\it y_1}={\frac {15}{16}}\,{\it \varrho}-3/16\,i\sqrt {3}{\it \varrho}\,,
 $$
 $$
   {\it y_2}=-3/16\,{\it \varrho}-{\frac {9}{16}}\,i\sqrt {3}{\it \varrho}\,,
 $$
 $$
   {\it y_3}=3/8\,{\it \varrho}+3\,i\sqrt {3}{\it \varrho}\,.
 $$
Obviously, this solution exists but it is not real and positive.
We may conclude that there exists no Hilbert-space metric which
would make both of the models (\ref{asyre}) and (\ref{cosy})
(with real spectra) mutually compatible, i.e.,
simultaneously quasi-Hermitian.
}

\textcolor{black}{
In conclusion let us add that
conceptually,
it would be entirely straightforward
to test also the compatibility of
the triplets or multiplets of certain
preselected operators $A_j$ with real spectra and with
$j = 1, 2, \ldots, K$.
In such cases we could merely replace
the two=term requirement (\ref{2veri})
by its $K-$term generalization.
The number of constraints will then grow quickly with $K$ since
the $(K-1)-$plet of Eqs.~(\ref{[15]})
(i.e., as many as $N(N-1)(K-1)/2$
formally independent conditions)
will have to be imposed
upon the mere $N-$plet of
variable parameters.
}

\section{Summary}

In our present paper we emphasized the importance of a
possible guarantee of a mutual compatibility
between any two different candidates $A_i$ and $A_2$
for the observables, and we
explained how
such a guarantee could be provided in practice.

Before the formulation of our results we emphasized that,
first of all, the
non-Hermiticity of $A_1\neq A_1^\dagger$ and of $A_2\neq A_2^\dagger$
is just an artifact of our
decision
of
performing
the mathematical operations (like, e.g.,
the localizations for spectra, etc)
in a preselected and, by assumption, manifestly
unphysical but by far the
user-friendliest Hilbert space ${\cal H}={\cal H}_{mathematical}$.
Secondly,
we pointed out that
there is a large difference between the quantum systems and
states living in the
finite- or infinite-dimensional Hilbert spaces and,
for the sake of brevity, we restricted our attention
just to the former family.

In the first step of our constructive
resolution of the problem
of compatibility of $A_1$ with $A_2$
we decided to treat
the inner-product metric
$\Theta$
as an operator
product
of Dyson map $\Omega$
with its conjugate
(cf. Eq.~(\ref{met}) above).

In the second step we noticed and
made use of the
close connection between the conjugate
operator  $\Omega^\dagger$
and an unnormalized
(or, more precisely, ``zeroth'', arbitrarily normalized)
set of vectors
defined as eigenvectors
of the conjugate
of any one of the observables in question.
This led quickly to our ultimate main result
(cf. Lemma \ref{lemmadva}).
Due to the generality of the latter result
%
\textcolor{black}{
we only have to remind}
the readers that
at present,
the QHQM-inspired
use of the Hermitizable operators is widely accepted
\cite{ali,Carl,SIGMA} and used
in methodical studies \cite{book,Carlbook}
as well as
in fairly diverse phenomenological considerations
\cite{BG,Christodoulides,KGali,Ju}.

In the narrower framework of the
hiddenly unitary quantum theory
a key to the mathematical consistency of
the corresponding description of physical reality
has been found
in the well known fact that a typical operator $A_j$
which is non-Hermitian in a preselected (and, presumably,
sufficiently ``user-friendly'')
Hilbert space ${\cal H}_{mathematical}$
can still be reinterpreted as self-adjoint in some other,
``amended''
Hilbert space ${\cal H}_{physical}$.
A decisive reason of feasibility
of such a model-building strategy
lies in the
possibility of amendment
${\cal H}_{mathematical}\to {\cal H}_{physical}$ of the space
which can be
mediated by the mere upgrade (\ref{change}) of the inner product in
${\cal H}_{mathematical}$
(to be compared with Eq. Nr. (2.2)
of review \cite{Geyer} where the authors used symbol $T$
in place of our present $\Theta$).

The latter form of the representation of
a user-unfriendly Hilbert space ${\cal H}_{physical}$
in its user-friendly alternative ${\cal H}_{mathematical}$
has been found to be a decisive technical merit of the
innovation.
The formalism could be called quantum mechanics in
quasi-Hermitian representation because
all of the candidates
$A_j$
for the observables had to satisfy the same quasi-Hermiticity
constraint (\ref{foral}) in
${\cal H}_{mathematical}$
(cf. also Eq. Nr. (2.1e) and the related comments
in \cite{Geyer}).

Without any danger of confusion one may
narrow the scope of our present considerations
by dealing with the first nontrivial scenario
in which one is given a quasi-Hermitian Hamiltonian $H \ (= A_1)$
and another observable $Q \ (= A_2)$ treated, very formally (and
in a way inspired by our preceding rather realistic
QHQM study \cite{arabky}), as
a ``spatial coordinate''.
A generalization to the models with multiplets of observables
is obvious.

Independently of the physical motivation as used in paper \cite{arabky}
the necessity of considering at least two operators of observables
emerged, most urgently, in the potential applications
of QHQM in quantum cosmology \cite{aliKG}.
In such a context, indeed, a mathematically consistent
treatment of quantum gravity necessarily requires
a nontrivial dynamical nature of the space.
In the related literature
the authors usually speak about the
``background-independent quantization''
of gravity -- see, e.g.,
the truly comprehensive monograph \cite{Thiemann}.

The complex nature of the problems of quantum gravity
lies already far beyond the scope of
our present paper, of course.
For our present methodical purposes
it was sufficient
to work just with the Hilbert spaces
of a finite dimension $N < \infty$, ${\cal H}= {\cal H}^{(N)}$.
Naturally, a transition to the models with $N=\infty$
might still be highly nontrivial: As a word of warning
let us recall, {\it pars pro toto},  a sample of
emergent subtleties in~\cite{Siegl} or \cite{Rovelli}.

\end{document}